\pdfoutput=1
\documentclass{llncs}
\usepackage{latexsym}
\usepackage{amssymb}
\usepackage{amsmath}
\usepackage{amssymb}
\usepackage{amsfonts}
\usepackage{epsfig}
\usepackage{color}
\usepackage{times}
\usepackage{graphics}
\usepackage{comment} 
\usepackage{verbatim}

\usepackage[utf8]{inputenc}
\usepackage[T1]{fontenc}

\usepackage{pdflscape}

\usepackage{graphicx}

\usepackage{pst-pdf}

\usepackage[margin=1.6in]{geometry}
\usepackage[]{algorithm2e}
\usepackage{array} 

\newcommand{\taskset}{\ensuremath{\mathcal{T}}}
\newcommand{\nvipa}{\ensuremath{k}}

\newcommand{\abs}[1]{\ensuremath{\lvert #1 \rvert}}

\newcommand{\tourd}{\ensuremath{\Gamma}}

\newcommand{\capV}{\ensuremath{\text{Cap}}}

\newcommand{\circuit}{\ensuremath{C}}
\newcommand{\subline}{\ensuremath{L}}
\newcommand{\uprequests}{\ensuremath{U}}
\newcommand{\downrequests}{\ensuremath{D}}
\newcommand{\IG}{\mathcal{G(I)}}
\newcommand{\I}{\mathcal{I}}
\newcommand{\X}{\mathcal{X}}
\newcommand{\OPTTRAM}{\textsc{Opt-Tram}\xspace}
\newcommand{\comp}{\ensuremath{c}}
\newcommand{\nsub}{\ensuremath{n}}
\newcommand{\kp}{\ensuremath{\ell}}
\newcommand{\nreq}{\ensuremath{m}}
\newcommand{\mods}[1]{{\color{blue}#1}}

\newcommand{\redd}[1]{{\color{blue}#1}}
\newcommand{\sah}[1]{{\color{blue}#1}}

\newcommand{\Z}{\mathbb{Z}}

\begin{document}
\begin{center}
\begin{LARGE}
\textbf{Fleet management for autonomous vehicles:\\ Online PDP under special constraints\\} 
\end{LARGE}
\end{center}

\begin{center}
Sahar Bsaybes, Alain Quilliot and Annegret K.~Wagler\\
{Universit\'e Clermont Auvergne \\
(LIMOS UMR CNRS 6158), Clermont-Ferrand, France}  

\end{center}

\begin{abstract}
The VIPAFLEET project consists in developing models and algorithms for managing a fleet of Individual Public Autonomous Vehicles (VIPA). Hereby, we consider a fleet of cars distributed at 
specified stations in an industrial area to supply internal transportation, where the cars can be used in different modes of circulation (tram mode, elevator mode, taxi mode). 
One goal is to develop and implement suitable algorithms for each mode in order to satisfy all the requests 
under an economic point of view 
by minimizing the total tour length. 
The innovative idea and challenge of the project is 
to develop and install a dynamic fleet management system that allows the operator to switch between the different modes within the different periods of the day according to the dynamic transportation demands of the users.
We model the underlying online transportation system and propose a corresponding fleet management framework, 
to handle modes, demands and commands. 
We consider two modes of circulation, tram and elevator mode, propose for each mode appropriate online algorithms and evaluate their performance, both in terms of competitive analysis and practical behavior. 
\end{abstract}

\section{Introduction}

The project VIPAFLEET aims at contributing to sustainable mobility through the development of innovative urban mobility solutions by means of fleets of Individual Public Autonomous Vehicles (VIPA) 
allowing passenger
transport in 
closed sites like industrial areas, medical complexes, campuses, business centers, big parkings, airports and train stations. 
A VIPA is an ``autonomous vehicle'' that does not require a driver nor an infrastructure to operate,
it is developed by Easymile and Ligier \cite{easymile,ligier} thanks to innovative computer vision  guidance  technologies developed by  researchers  at  Institut  Pascal~\cite{royer2005outdoor,RFIA2016}.
This innovative project involves different partners in order to ensure the reliability of the transportation system \cite{viameca}. A long-term experimentation~\cite{RFIA2016} has been performed on the 
industrial site ``Ladoux'' of Michelin at Clermont-Ferrand as part of the FUI VIPAFLEET project.
A fleet of VIPAs shall be used in an industrial site 
to transport employees and visitors e.g. between parkings,
buildings and from or to the restaurant at lunch breaks.
The fleet is distributed at specified stations 
to supply internal transportation, and a VIPA can operate in three different transportation modes:
\begin{itemize}
\item \emph{Tram mode:} VIPAs continuously run on predefined lines or cycles in a predefined direction and stop at a station if requested to let users enter or leave. 
\item \emph{Elevator mode:} VIPAs run on predefined lines 
and react to customer requests by moving to a station to let users enter or leave, thereby changing their driving direction 
if needed. 
\item \emph{Taxi mode:} users book their transport requests (from any start to any destination station within the network with a start and an arrival time) in real time. 
\end{itemize}

This leads to a Pickup-and-Delivery Problem (PDP) where a fleet of servers shall transport goods or persons from a certain origin to a certain destination.  
If persons have to be transported, we usually speak about a Dial-a-Ride Problem (DARP). 
Many variants 
are studied 
including the Dial-a-Ride Problem with time windows \cite{deleplanque2013transfers,fabri2007online}. 
In our case, we are confronted with an online situation, where transportation requests 
are released over time \cite{ascheuer2000online,berbeglia2010dynamic,cordeau2007dial}.

In a practical context, different objective functions can be considered. 
We focus on the economic aspect where the objective is to minimize costs by minimizing the total tour length.

In a VIPAFLEET system, users can either call a VIPA directly from a station with the help of a call-box 
or can book their request in real time by mobile or web applications.
Since the  customer requests are released over time, the classical PDP does not reflect the real situation of this project.
Therefore we are interested in its online version. 
Our aim is to develop and install a Dynamic Fleet Management System that allows the operator to switch between different network designs, transportation modes and appropriate online algorithms within the different periods of the day in order to react to changing demands 
evolving during the day, with the objective to satisfy all demands in a best possible way.
For that we model the underlying online transportation system and propose an according fleet management framework, 
to handle modes, demands and commands (see Section 2). In Section 3 we
develop suitable online algorithms for tram and elevator mode and analyse them in terms of competitive analysis w.r.t minimizing the total tour length. 
In Section 4, we provide computational results showing that the practical behavior of the algorithms is much better than their worst case behavior captured by their competitive factors. This enables us to cluster the demands 
into subproblems in such a way that, for each subproblem, a suitable subnetwork and a suitable algorithm can be proposed leading to a globally good solution (transportation schedule).\\

\section{Problem description and model}
We embed the VIPAFLEET management problem in the framework of a metric task system.
We encode the closed site  where the VIPAFLEET system is running as a \emph{metric space} $M =(V,d)$ induced by a connected network $G = (V,E)$, where 
the nodes correspond to stations, arcs to their physical links in the closed site, and the distance $d$ between two nodes $v_i,v_j \in V$ is the length of a shortest path from $v_i$ to $v_j$. 
In $V$, we have a distinguished origin $v_o \in V$, the depot of the system where all $\nvipa$ VIPAs are parked when the system is not running, i.e., outside a certain time horizon $[0,T]$.

A Dynamic Fleet Management System shall allow the operator to switch between different transportation modes within the different periods of the day in order to react to changing customer demands 
evolving during the day. 
For that, for a certain period $[t,t']\subseteq[0,T]$, we define a 
metric subspace $M' =(V',d')$ induced by a subnetwork $G' = (V',E')$ of $G$, 
where a subset of nodes and arcs of the network is active (i.e. where the VIPAs have the right to perform a move on this arc or pass by this node during $[t,t']$).  
An operator has to decide when and how to move the VIPAs in the subnetworks, and to assign customer requests to VIPAs.

Any customer request $r_j$ is defined as a 4-tuple \sah{$r_j=(t_j,x_j,y_j,z_j)$} where  
\begin{itemize}
\item $t_j \in [0,T]$ is the release date, 
\item $x_j \in V$ is the origin node, 
\item $y_j \in V$ is the destination node, 
\item $z_j$ specifies the number of passengers. 
\end{itemize}

The operator monitors the evolution of the requests and the movement of VIPAs over time and creates tasks to move the VIPAs to go to some station and pick up, transport and deliver users.
More precisely, a \emph{task} can be defined as 
$\tau_j = (t_j,x_j,y_j,z_j,G^\prime)$. This task is sent at time $t_j$ to a VIPA operating in the subnetwork $G^\prime$ containing stations $x_j$ and $y_j$, indicating that $z_j$ passengers have to be picked up at $x_j$ and delivered at $y_j$.
In order to fulfill the tasks, we let a fleet of VIPAs (one or many, each having a capacity for $\capV$ passengers) circulate in the network inducing the metric space. 

More precisely we determine a feasible transportation schedule $S$ for $(M, \taskset)$ consisting of a collection of tours $\{\tourd^1, \ldots, \tourd^\nvipa \}$ where:
\begin{itemize}
 \item\label{def: enum: schedule: 1} each of the $\nvipa$ VIPAs has exactly one tour,
 \item\label{def: enum: schedule: 2} each request  $r_j$  is served not earlier than the time $t_j$ it is released,
 \item each tour starts and ends in the depot.
\end{itemize}

If each user is transported from its start station to its final destination by only one VIPA, then $S$ is called non-preemptive, otherwise preemptive. 
In particular, if start and destination station of a user do not lie on the same subnetwork $G^\prime$, the user has to change VIPAs on one or more intermediate stations. 
As this typically leads to inconveniences, the design of subnetworks has to be done in such a manner that preemption can be mostly avoided. 

In addition, depending on the policy of the operator of such a system, different technical side constraints have to be obeyed.
If two or many VIPAs circulate on the same (sub)network, the fleet management has to handle e.g. the 
\begin{itemize}
 \item meeting of two vehicles on a station or an arc,
 \item blocking the route of a VIPA by another one waiting at a station (if two VIPAs are not allowed to enter the same node or arc at the same time),
\end{itemize}
and has to take into account 
\begin{itemize}
 \item the events of breakdown or discharge of a vehicle,
 \item technical problems with the server, the data base or the communication network between the stations, VIPAs and the central server. 
\end{itemize}
Our goal is to construct transportation schedules $S$ for the VIPAs respecting all the above constraints w.r.t minimizing the total tour length.

This will be addressed by dividing the time horizon $[0,T]$ in different periods according to the volume and kind of the requests, and by providing specific solutions within each period.
The global goal is to provide a feasible transportation schedule over the whole time horizon that satisfies all requests and minimizes the total tour length (Dynamic Fleet Management Problem).
For each period, we partition the network $G = (V,E)$ 
into a set of subnetworks $G^\prime = (V^\prime,E^\prime)$.
The aim of this partition is 
to solve some technical side constraints for autonomous vehicles, 
and to solve the online PDP using different algorithms at the same time on different subnetworks. To gain precision 
in solutions by applying a suitable algorithm to a certain subnetwork, 
the choice of the design of the network in the industrial site where the VIPAs will operate is dynamic and will change over time according to the technical features and properties.

In a metric 
space we partition the network 
into subnetworks $G^\prime$ that are either unidirected cycles, called \textit{circuits} $G^\prime_c$, or bidirected paths, called \textit{lines} $G^\prime_{\ell}$. This partition is motivated by two of the operation modes 
for VIPAs, tram mode and elevator mode, that we will consider in the sequel of the paper.

\section{Scenarios: combinations of modes and subnetworks}
Based on all the above technical features and properties that have an impact on the feasibility of the transportation schedule, we can cluster the requests into subproblems, apply to each subproblem a certain algorithm, and check the results in terms of feasibility and performance.

The choice of the design of the network in the industrial site  where the VIPAs will operate is dynamic and will change over time according
to the technical features and properties. We consider four typical scenarios that occurred while operating a fleet in an industrial site\footnote{A long-term 
experimentation~\cite{RFIA2016} has been performed in the industrial site  ``Ladoux'' of Michelin at Clermont-Ferrand where
a fleet of VIPAs was operating for several months (October 2014 - February 2015).}
based on some preliminary studies of the transport requests within the site.
\paragraph{Morning/evening:} 
The transport requests are between parkings and buildings. 
For this time period we propose the following:
 \begin{itemize}
  \item [$\bullet$] Design a collection of subnetworks (lines and circuits) s.t.
  \begin{itemize}
   \item [-] all buildings and parkings are covered,
   \item [-] each subnetwork contains one parking $p$ and all the buildings where $p$ is the nearest parking (to ensure that for each request, origin (the parking) and destination (a building) lie in the same subnetwork).
  \end{itemize}
  \item [$\bullet$] Depending on the number of employees in the served buildings, assign one VIPA (in elevator mode) to every line and one or several VIPAs (in tram mode) to each circuit.
 \end{itemize}

\paragraph{Lunch time:}
The transport requests are between buildings and the restaurant of the industrial complex. 
For this time period, we propose the following:
 \begin{itemize}
  \item [$\bullet$] Design a collection of lines s.t.
  \begin{itemize}
   \item [-] all buildings are covered,
   \item [-] each line contains the station of the restaurant (to ensure that for each request, to or from the restaurant, origin and destination lie in the same line).
  \end{itemize}
  \item [$\bullet$] Depending on the number of employees in the served buildings, assign one VIPA (in elevator mode) or one or several VIPAs (in tram mode) to the lines.
 \end{itemize} 
\paragraph{Emergency case:} 
In the case of a breakdown of the central servers, the database or the communication system, transports between all possible origin/destination  pairs have to be ensured without any decision by the operator. For that we propose
 \begin{itemize}
  \item to use one Hamilton cycle through all the stations as subnetwork and 
  \item to let half of the fleet of VIPAs operate in each direction on the cycle (all in tram mode).
 \end{itemize}
\paragraph{Other periods:} 
There are mainly unspecified requests without common origins or common destinations. 
 The operator can use all VIPAs in his fleet in taxi mode on the complete network or design lines and circuits s.t. all stations are covered and the chosen subnetworks intersect (to ensure transports between all 
 possible origin/destination pairs). E.g., 
this can be done by
 \begin{itemize}
  \item using one Hamilton cycle through all stations where half of the fleet operates (in tram mode) in each direction,
  \item a spanning collection of lines and circuits meeting in a central station where one VIPA (in elevator mode) operates on each line, one or several VIPAs (in tram mode) on each circuit.
 \end{itemize}
\begin{remark}
One of these combinations of subnetworks has been used in a  long  duration  
experimentation~\cite{RFIA2016} on the industrial site ``Ladoux'' of Michelin at Clermont-Ferrand.
VIPAs were operating 
using the tram mode on a unidirected circuit where some of the buildings and parkings of the industrial site were covered.
\end{remark}

\section{Online algorithms and competitive analysis}
Recall that the customer requests are released over time s.t. the studied PDP has to be considered in its online version. 
A detailed introduction to online optimization can be found e.g. in the book by Borodin and El-Yaniv \cite{borodin2005online}.
It is standard to evaluate the quality of online algorithms with the help of competitive analysis.
This can be viewed as a game between an online algorithm $ALG$ and a malicious adversary who tries to generate a
worst-case request sequence $\sigma$ which maximizes the ratio between the online cost $ALG(\sigma)$ and the optimal offline cost $OPT(\sigma)$
knowing the entire request sequence $\sigma$ in advance.
 $ALG$ is called \emph{$\comp$-competitive} if $ALG$ produces for any request sequence $\sigma$ a feasible solution 
 with
\[
 ALG(\sigma)  \leq \comp \cdot OPT (\sigma)
\]
for some given $\comp \geq 1$. The competitive ratio of ALG is the infimum over all $\comp$ such that $ALG$ is $\comp$-competitive.
Here, we are interested in designing and analyzing online algorithms for two possible operating modes of a VIPA, tram mode and elevator mode.
\subsection{Tram mode}
The tram mode is the most restricted operation mode where VIPAs run on predefined circuits in a predefined direction and stop at stations to let users enter or leave.
 We consider circuits $\circuit$ with one distinguished node, the origin of the circuit\footnote{Circuits can also be lines where one end is the origin and the VIPA can change its direction only in the two ends of the line.}.

\paragraph{Optimal offline solution via colorings of interval graphs:}
An interval graph $\IG$ is obtained as intersection of a set $\I$ of intervals within a line segment, where
\begin{itemize}
 \item the nodes of $\IG$ represent the intervals in $\I$,
 \item the edges of $\IG$ represent their conflicts in terms of overlaps (i.e. two nodes are adjacent if the corresponding intervals have a non-empty intersection).
\end{itemize}
The clique number $w(\IG)$ corresponds to the largest number of pairwise intersecting intervals in $\I$, a coloring corresponds to an assignment of colors to intervals such that no two intersecting intervals receive the same color. 
In all graphs, the clique number is a lower bound on the minimum number $\X(\IG)$ of required colors. For interval graphs it was shown in~\cite{olariu1991optimal} that the following Greedy coloring algorithm 
always produces an $w(\IG)$-coloring of $\IG$:
\begin{itemize}
 \item sort all intervals in $\I$ according to their left end points.
 \item color the nodes of $\IG$ in this order: starting with the first node, assign to each node the smallest color that none of its already colored neighbors has.
\end{itemize}

We next interpret the offline solution for VIPAs operating in tram mode on a circuit in this context.
We have given a circuit $\circuit=\{v_o,v_1, \ldots, v_{\nsub}\}$ and a sequence $\sigma$ of $\nreq$ requests $r_j=(t_j,x_j,y_j,z_j)$ with origin/destination pairs $(x_j,y_j)\in \circuit \times \circuit$. 
W.l.o.g. we may assume that the origin $v_0$ of the circuit does not lie in the interior of $(x_j,y_j)$ for any $r_j \in \sigma$.
We transform $\circuit$ into a path $P=\{v_0,v_1,\ldots , v_{\nsub}, v_0\}$ having the origin $v_0$ of $\circuit$ as start and end node (as the line segment), and we split each request $r_j$ into $z_j$
many uniform requests (resp. single passengers), interpreted as subpaths $(x_j,y_j) \subseteq P$ (to obtain the (multi) set $\I$ of intervals).
By construction, we have for the resulting interval graph $G_{\sigma}=\IG$:
\begin{itemize}
 \item the clique number $w(\IG)$ corresponds to the maximum number of requests $r_j$ in $\sigma$ (counted with their multiplicities $z_j$) traversing a same edge of $P$,
 \item a coloring of $G_{\sigma}$ corresponds to an assignment of places in the VIPA(s) to passengers.
\end{itemize}

Clearly one VIPA can serve all (uniform) requests from up-to $\capV$ color classes in a single subtour traversing $\circuit$. We can, thus, turn any coloring of $G_{\sigma}$ into a feasible transportation schedule by
\begin{itemize}
 \item waiting until time $t_{\nreq}$ in the origin $v_0$ ( to ensure that all requests are released before they are served)
 \item selecting up to $\capV$ many color classes and assigning the corresponding uniform requests (i.e. single passengers) to one VIPA, to be served within the same subtour traversing $\circuit$, until all requests are served.
\end{itemize}
This leads to the following algorithm to compute optimal offline solutions for the tram mode: \\ 
\OPTTRAM\ \\
Input: $\sigma=\{r_1,r_2,\ldots, r_{\nreq}\}$, $\circuit=\{v_o,v_1, \ldots, v_{\nsub}\}$, $\capV$ and $\nvipa$\\
Output: transportation schedule
\begin{enumerate}
 \item for each $r_j=(t_j,x_j,y_j,z_j) \in \sigma$: create $z_j$ many intervals $(x_j,y_j)$ to obtain $\I$,
 \item sort all intervals in $\I$ according to their left end points,
 \item create the interval graph $\IG$ and apply the Greedy algorithm to color it,
 \item wait until $t_{\nreq}$ (the release date of the last request),
 \item as long as there are unserved requests:
 \item []select $\capV$ many (or all remaining) color classes, assign the corresponding passengers to a VIPA and perform one subtour traversing $\circuit$ to serve them.
\end{enumerate}

\begin{example}
\label{exa: opt: tram-ttl}

 Consider a circuit $\circuit=(a,b,c,d,e)$ with origin $a$ and one unit-speed server with capacity \mbox{$\capV=2$} 
(i.e. a VIPA that travels 1 unit of length in 1 unit of time), and  a sequence $\sigma$ of $6$ requests:
  \[r_1=(1,c,e,2)\hspace{0.5cm} r_4=(4,b,c,2)\]\\[-1cm]
  \[r_2=(2,a,d,1)\hspace{0.5cm} r_5=(5,a,b,1)\]\\[-1cm]
  \[r_3=(3,d,e,1)\hspace{0.5cm} r_6=(6,b,e,1)\]\\[-6mm]
\newpage

\begin{enumerate}
\item for each $r_j \in \sigma$ we create $z_j$ many intervals $(x_j,y_j)$ to obtain $\I$:
\vspace{-0.5cm}
\begin{figure}[ht]
    \centering
    \includegraphics[width=0.7\textwidth]{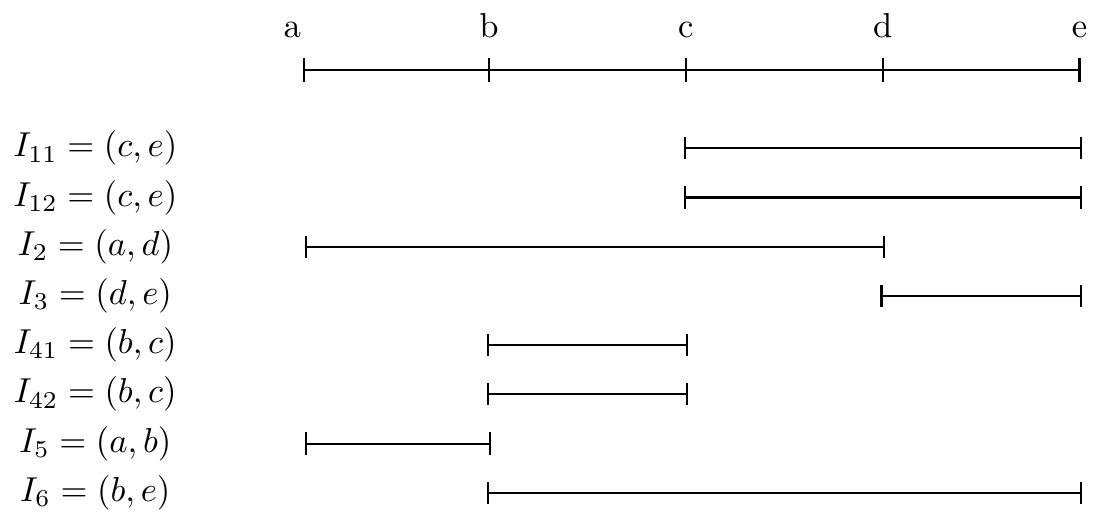}
 \label{fig: opt-tram-interval-coloring}
\end{figure}
\item We sort all intervals in $\I$ according to their left end points:
$I_5, I_2, I_{41}, I_{42}, I_6, I_{11}, I_{12}, I_3$.
\item We create the interval graph $\IG$: 
\begin{figure}[ht]
    \centering
    \includegraphics[width=0.7\textwidth]{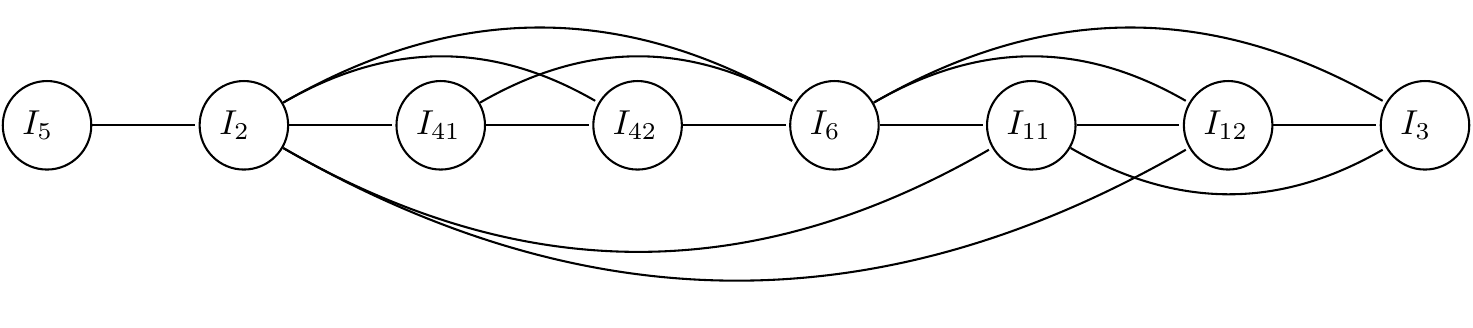}
 \label{fig: opt-tram-interval-graph}
 \vspace{-0.5cm}
\end{figure}
\item We apply the Greedy algorithm to color it: 
\begin{figure}[ht]
    \centering
    \includegraphics[width=0.7\textwidth]{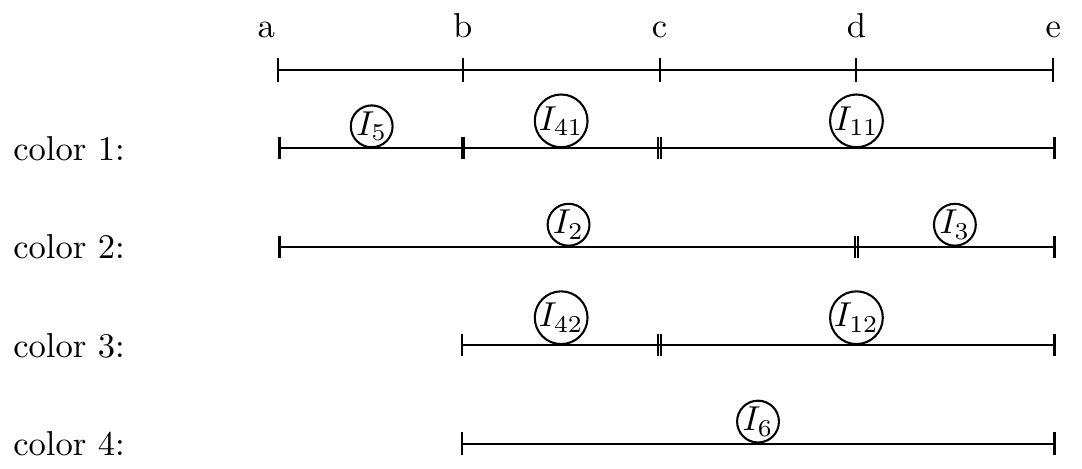}
 \label{fig: opt-tram-greedy-coloring}
\end{figure}
\item As long as there are unserved requests, we select 2 random color classes, assign the corresponding passengers to a VIPA and perform one subtour traversing $\circuit$ to serve them, for instance: 
\begin{itemize}
 \item first VIPA, first round: color 1 and 2 ($r_5,r_2,r_{41},r_{11},r_3$)
 \item second VIPA, or second round of first VIPA: color 3 and 4 ($r_{42},r_{12},r_6$)
\end{itemize}
\end{enumerate}
\end{example}
\begin{theorem}
 Algorithm \OPTTRAM\ provides a load-preemptive optimal offline solution w.r.t. minimizing the total tour length for VIPAs operating in tram mode on a circuit.
\end{theorem}

\begin{proof}
 By construction, splitting the requests $r_j$ from $\sigma$ according to their multiplicities $z_j$ gives a (multi)set $\I$ of intervals where each of them stands for a uniform request of a single passenger. 
 Accordingly, the clique number $w(\IG)$ of the resulting interval graph $\IG$ corresponds to the maximum number of passengers traversing a same edge $e$ of the circuit $\circuit$, that is
 \[
  w(\IG)= max\{\sum_{e\in (x_j,y_j),r_j \in \sigma} z_j;  e\in \circuit \}
 \]
On the other hand, a coloring of $\IG$ corresponds to an assignment of places in the VIPA(s) to passengers, and subtours for the VIPA(s) can be easily created by repeatedly choosing $\capV$ color classes and 
assigning the passengers colored that way to the $\capV$ places of one VIPA. Clearly, at least $\left \lceil { \frac{w(\IG)}{\capV} } \right\rceil$ many such subtours are needed to serve all requests.
The transportation schedule obtained is feasible because, by waiting until $t_{\nreq}$ to start any subtour, we ensure that all requests have been released before.
As the Greedy coloring algorithm provides an optimal $w(\IG)$-coloring of $\IG$, we can guarantee to obtain a feasible transportation schedule performing the minimal number of subtours by always choosing $\capV$ colors
(except for the last subtour where we choose all remaining ones) so that the minimal total tour length equals
\[
 OPT(\sigma)=\left \lceil { \frac{w(\IG)}{\capV} } \right\rceil \cdot \abs{\circuit}.
\]
The resulting solution is a load-preemptive transportation schedule because it cannot be ensured that all $z_j$ passengers coming from the same request $r_j$ are served by the same VIPA (even if $z_j \leq \capV$ holds).
\hfill$\square$
\end{proof}

\begin{remark}\hspace{0.2cm}\\[-0.7cm]
\begin{enumerate}
 \item [$\bullet$]The minimal total tour length does not depend on the number of VIPAs used to serve all requests.
 \item [$\bullet$]Algorithm \OPTTRAM\ is clearly polynomial because  all the steps of the algorithm can be computed in polynomial time. 
 \item [$\bullet$]By not selecting $\capV$ color classes randomly to create subtours, it is possible to:
  \subitem reduce load-preemption,
  \subitem minimize the number of stops performed to let passengers leave/enter a VIPA,
  \subitem handle the case of more than one VIPA,\\
  but the so modified algorithm is not necessarily polynomial anymore.
\end{enumerate}
\end{remark}

\paragraph{Online algorithms:}
For the online situation, we propose the following simple algorithm for VIPAs operating in tram mode on a circuit $\circuit$:\\

\noindent
SIR (``Stop If Requested'')\\[-6mm] 
\begin{itemize}
  \item each VIPA waits in the origin of $\circuit$; 
as soon as a request is released, a VIPA starts a full subtour in a given direction, thereby it stops at a station when a user requests to enter/leave.
 \end{itemize}
In fact, in tram mode, the possible decisions of the VIPA are either to continue its tour or to wait at its current position for newly released requests.
This can be used by the adversary to ``cheat'' an online algorithm, in order to maximize the ratio between the online and the optimal costs.
Here, the strategy of the adversary is to force SIR to serve only one uniform request per subtour, whereas the adversary only needs a single subtour traversing $\circuit$ to serve all requests.

\begin{example}
\label{exa: comp: SIR: CAPC-totaltourlength}
Consider a circuit $\circuit=(v_0,v_1,\dotsc,v_{\nsub})$ with origin $v_0$, a unit distance between $v_i$ and $v_{i+1}$ for each $i$, and one unit-speed server with capacity $\capV$.
The adversary releases  a sequence $\sigma$ of $\capV \cdot n$ uniform requests that force SIR to perform one full round (subtour) of length $\abs{C}=\nsub+1$ for each uniform request, 
whereas the adversary is able to serve all requests in a single subtour (fully loaded on each edge):
\begin{itemize}
 \item $\capV$ requests $r_j=((j-1)\abs{C},v_0,v_1,1)$ for $1 \leq j \leq \capV$
 \item $\capV$ requests $r_j=((j-1)\abs{C},v_1,v_2,1)$ for $\capV+1 \leq j \leq 2\capV$\\
 $\vdots$
 \item $\capV$ requests $r_j=((j-1)\abs{C},v_{\nsub-1},v_{\nsub},1)$ for $(\nsub-1)\capV+1 \leq j \leq \nsub\capV$
 \item $\capV$ requests $r_j=((j-1)\abs{C},v_{\nsub},v_0,1)$ for $\nsub\capV+1 \leq j \leq (\nsub+1)\capV$
\end{itemize}
SIR starts its VIPA at time $t=0$ to serve $r_1=(0,v_0,v_1,1)$ and finishes the first subtour of length $\abs{C}$ without serving any further request.
When the VIPA operated by SIR is back to the origin $v_0$, the second request $r_2=(\abs{C},v_0,v_1,1)$ is released and SIR starts at $t=\abs{C}=\nsub+1$ a second subtour of length $\abs{C}$ to serve $r_2$,
without serving any further request in this subtour. This is repeated for each request yielding
 $ SIR(\sigma)=\capV \cdot \abs{C} \cdot \abs{C}.$
 
The adversary waits at the origin $v_1$ until $t=(\capV -1)\abs{C}$ and serves all requests $r_1,\dotsc,r_{\capV}$ from $v_0$ to $v_1$. Then he waits until $t=(2\capV -1)\abs{C}$ at $v_1$ and serves all requests 
$r_{\capV+1},\dotsc, r_{2\capV}$ from $v_1$ to $v_2$. This is repeated for all $\capV$ requests from $v_i$ to $v_{i+1}$, yielding 
$OPT(\sigma)=\abs{C}.$
The tours performed by SIR and OPT are illustrated in Fig~\ref{fig: SIR-worst-case-ten}.
Therefore, we obtain \[\frac{SIR(\sigma)}{OPT(\sigma)} = \frac{\capV \cdot \abs{C} \cdot \abs{C}}{\abs{C}}= \capV \cdot \abs{C}\]
as a lower bound for the competitive ratio of SIR.
\begin{figure}[ht]
    \centering
    \includegraphics[width=0.99\textwidth]{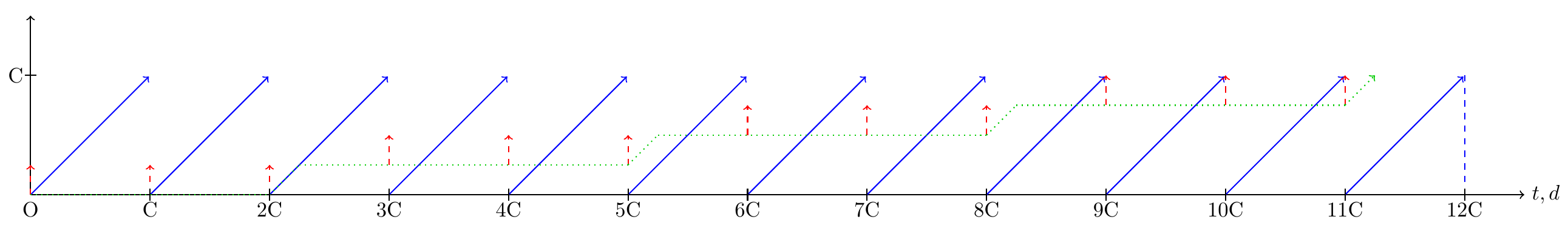}
 \caption{This figure illustrates the tour performed by SIR (in blue) and the adversary (dotted in green) in order to serve the requests (dashed arcs in red) from Example~\ref{exa: comp: SIR: CAPC-totaltourlength} for $\capV = 3, \nsub=3$ and $\abs{C}=4$.}
 \label{fig: SIR-worst-case-ten}
\end{figure}
\end{example}
In the special case of the lunch scenario, we may consider VIPAs operating in tram mode on circuits, where each circuit has the restaurant as its distinguished origin. 
A sequence $\sigma'$ containing the first $\capV$ and the last $\capV$ requests from the sequence presented in Example~\ref{exa: comp: SIR: CAPC-totaltourlength} shows that $2 \capV$ is a lower bound on the competitive ratio of SIR, 
see Figure~\ref{fig: SIR-ttl-lunch-worst-case} for an illustration.
\begin{figure}[ht]
    \centering
    \includegraphics[width=0.99\textwidth]{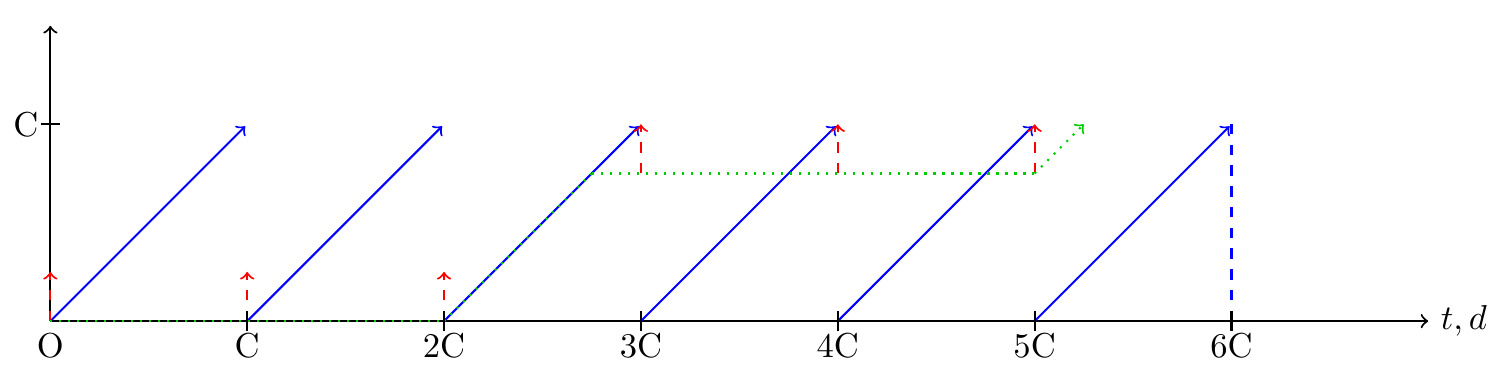}
 \caption{This figure illustrates the tour performed by SIR (in blue) and the adversary (dotted in green) in order to serve the first $\capV$ and the last $\capV$ requests (dashed arcs in red) 
 from the sequence presented in Example~\ref{exa: comp: SIR: CAPC-totaltourlength} for $\capV = 3, \nsub=3$ and $\abs{C}=4$.  These requests satisfy the criterias of the lunch scenario.}
 \label{fig: SIR-ttl-lunch-worst-case}
\end{figure}

As for the morning resp. evening scenario, we consider VIPAs operating in tram mode on a circuit $\circuit$ where the parking is the distinguished origin of $\circuit$. 
A sequence $\sigma''$ containing the first $\capV$ resp. last $\capV$ requests from the sequence presented in Example~\ref{exa: comp: SIR: CAPC-totaltourlength} shows that $\capV$ is a lower bound on the competitive ratio of SIR, 
see Figure~\ref{fig: SIR-ttl-morning-worst-case} and Figure~\ref{fig: SIR-ttl-evening-worst-case}, respectively.\\
\begin{figure}[ht]
    \centering
    \includegraphics[width=0.6\textwidth]{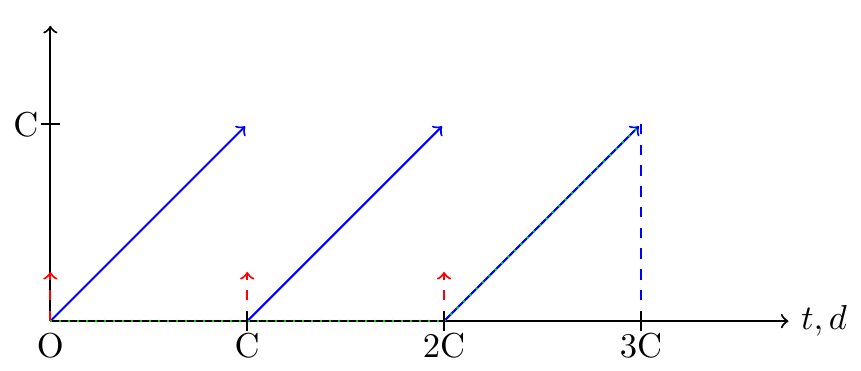}
 \caption{This figure illustrates the tour performed by SIR (in blue) and the adversary (dotted in green) in order to serve the first $\capV$ requests (dashed arcs in red)   
 from the sequence presented in Example~\ref{exa: comp: SIR: CAPC-totaltourlength} for $\capV = 3, \nsub=3$ and $\abs{C}=4$. 
 These requests satisfy the criterias of the morning scenario.}
 \label{fig: SIR-ttl-morning-worst-case}
\end{figure}
\begin{figure}[ht]
    \centering
    \includegraphics[width=0.6\textwidth]{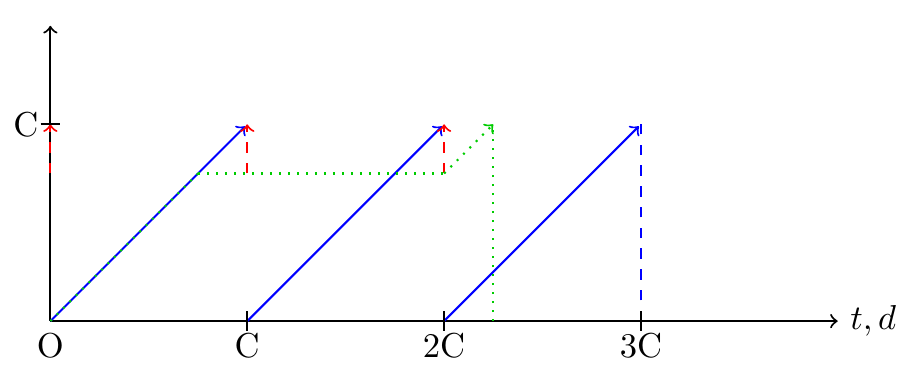}
 \caption{This figure illustrates the tour performed by SIR (in blue) and the adversary (dotted in green) in order to serve the last $\capV$ requests (dashed arcs in red)  
 from the sequence presented in Example~\ref{exa: comp: SIR: CAPC-totaltourlength} for $\capV = 3, \nsub=3$ and $\abs{C}=4$. These requests satisfy the criterias of the evening scenario.}
 \label{fig: SIR-ttl-evening-worst-case}
\end{figure}

\newpage
We can prove that the previously presented examples are indeed worst cases for SIR:
\begin{theorem}\label{th-SIR-ttl}
For one or several VIPAs with capacity $\capV$ operating in tram mode on a circuit $\circuit$ with length $\abs{C}$,
SIR is  w.r.t the objective of minimizing the total tour length 
\begin{itemize}
 \item $\capV \cdot \abs{C}$-competitive in general,
 \item $2 \cdot \capV$-competitive  during the lunch scenario,
 \item $\capV$-competitive during the morning scenario resp. the evening.
\end{itemize}
\end{theorem}

\begin{proof}
 Recall that a transportation schedule is based on a coloring of the interval graph $G_{\sigma}$, whose nodes stand for passengers from $\sigma$, i.e. to the requests
 $r_j \in \sigma$ counted with their multiplicities $z_j$.
 The worst coloring of $G_{\sigma}$ is to assign different colors to all nodes, i.e. using $\abs{G_{\sigma}}= \sum_{r_j\in \sigma} z_j$ many colors.
 The worst transportation schedule results if, in addition, each VIPA performs a separate subtour of length $\abs{\circuit}$ for each color (i.e. serving a single uniform request only per subtour), 
 yielding $\abs{G_{\sigma}}\cdot \abs{\circuit}$ as total tour length.\\ 
 
 SIR can indeed be forced to show this behavior by releasing the requests accordingly (i.e. by using uniform requests with $z_j=1$ each and with sufficiently large delay between 
  $t_j$ and $t_{j+1}$),
 \begin{itemize}
  \item in general: using the sequence $\sigma$ from Example~\ref{exa: comp: SIR: CAPC-totaltourlength},
  \item during lunch: using the sequence $\sigma^\prime$ restricted to the first $\capV$ and the last $\capV$ requests $(t_j,v_0,v_1,1)$ and $(t_j,v_{\nsub},v_0,1)$ from the sequence  $\sigma$ presented
  in Example~\ref{exa: comp: SIR: CAPC-totaltourlength} as \sah{in} Figure~\ref{fig: SIR-ttl-lunch-worst-case},
  \item during morning/evening: using the sequence $\sigma''$ restricted to the first $\capV$ requests $(t_j,v_0,v_1,1)$ (resp. the last $\capV$ requests $(t_j,v_{\nsub},v_0,1)$) from the sequence $\sigma$ presented in Example~\ref{exa: comp: SIR: CAPC-totaltourlength}, as 
  Figure~\ref{fig: SIR-ttl-morning-worst-case} (resp. Figure~\ref{fig: SIR-ttl-evening-worst-case}) shows.
 \end{itemize}
Furthermore, to maximize the ratio between this total tour length obtained by SIR and the optimal offline solution, we need to ensure that all requests in $\sigma$ can be served with as few subtours of length $\abs{\circuit}$ as possible.
This is clearly the case if all requests have length 1 and there are $\capV$ many requests traversing the same edge of $\circuit$ s.t. a single subtour suffices to serve all of them
(see again Example~\ref{exa: comp: SIR: CAPC-totaltourlength}).
This leads to 
\begin{itemize}
 \item $\abs{G_{\sigma}}=\abs{\sigma}=\capV \cdot \abs{\circuit}$  and $w(\IG)=\capV$ s.t.
 \[
 \frac{SIR(\sigma)}{OPT(\sigma)}= \frac{\capV \cdot \abs{\circuit} \cdot \abs{\circuit}}{1 \cdot \abs{\circuit}}= \capV \cdot \abs{\circuit}
\]
is the maximum possible ratio between $SIR(\sigma)$ and $OPT(\sigma)$ taken over all possible sequences in general.
\item  $\abs{G_{\sigma'}}=\abs{\sigma'}=2\capV$  and $w(\IG)=\capV$ s.t.
\[
 \frac{SIR(\sigma')}{OPT(\sigma')}= \frac{2\cdot \capV \cdot \abs{\circuit}}{1 \cdot \abs{\circuit}}= 2\cdot \capV
\]
is the maximum possible ratio between $SIR(\sigma)$ and $OPT(\sigma)$ taken over all possible sequences during the lunch.
\item$\abs{G_{\sigma''}}=\abs{\sigma''}=\capV$  and $w(\IG)=\capV$ s.t.
\[
 \frac{SIR(\sigma'')}{OPT(\sigma'')}= \frac{\capV \cdot \abs{\circuit}}{1 \cdot \abs{\circuit}}=\capV
\]
is the maximum possible ratio between $SIR(\sigma)$ and $OPT(\sigma)$ taken over all possible sequences during the morning or evening.
\hfill$\square$
\end{itemize}
\end{proof}

Moreover, SIR can be adapted to follow the strategy of the adversary. 
For that, we propose two other algorithms for VIPAs operating in tram mode in the morning resp. evening:
\newpage
\noindent
$SIF_M$ (``Start if fully loaded'') for the morning scenario\\[-6mm]  
\begin{itemize}
 \item each VIPA waits in the parking until $\capV$ passengers have entered,
 \item it starts a full round (as soon as it is fully loaded) and stops at stations where passengers request to leave.
\end{itemize}

\noindent
$SIF_E$ (``Start if fully loaded'') for the evening scenario\\[-6mm]  
\begin{itemize}
 \item each VIPA waits in the parking until enough requests are released to reach $\capV$,
 \item it starts a full round and stops at stations where passengers request to enter and returns (fully loaded) to the parking.
\end{itemize}
We can ensure optimality for these two strategies: 
\begin{theorem}
$SIF_M$ (resp. $SIF_E$) is $1$-competitive w.r.t minimizing the total tour length for one or several VIPAs operating in tram mode on a
circuit during the morning (resp. evening).
\end{theorem}
\begin{proof}
 Both variants of $SIF$ are optimal, because due to the special request structure during the morning resp. evening, all requests traverse the first (resp. last) edge of $P$ in the morning (resp. evening) s.t. $G_{\sigma}$ becomes a clique.
 In other words, no two passengers can share a same place in a VIPA s.t. starting fully loaded from the origin (resp. returning fully loaded to the origin) indeed provides the optimal solution w.r.t. minimizing the total tour length.
 \hfill$\square$
\end{proof}

The two previous algorithms $SIF_E$ and $SIF_M$ can be merged together to obtain a version for the lunch scenario:\\

\noindent
$SIF_L$ (``Start if fully loaded on at least one arc'')\\[-6mm]  
\begin{itemize}
 \item each VIPA waits in the restaurant until enough requests are released 
 s.t. by serving these requests using one VIPA, the VIPA is fully loaded at least on one arc in the circuit.
 \item it starts a full round and stops at stations where passengers request to enter or leave and returns to the restaurant.
\end{itemize}

\begin{theorem}
$SIF_L$ is $2$-competitive w.r.t minimizing the total tour length for one or several VIPAs operating in tram mode on a
circuit during the lunch.
\end{theorem}
\begin{proof}
Due to the special request structure during the lunch, all requests start and end in the restaurant and, thus, traverse the first and the last edge of $P$
s.t. $G_{\sigma}$ consists of two cliques $Q_1$ and $Q_2$ resulting from all uniform requests
traversing the first and the last edge, respectively.\\
The worst transportation schedule of $SIF_L$ results if the requests are released in a way that $SIF_L$ never serves a request from $Q_1$ with one from $Q_2$ together,
therefore by each subtour of length $\abs{\circuit}$ performed by $SIF_L$ $\capV$ requests are served either from the clique $Q_1$ or from $Q_2$,
yielding $SIF_L(\sigma)=\lceil\frac{\abs{G_\sigma}}{\capV}\rceil \cdot \abs{\circuit}$ as total tour length. \\
In order to maximize the ratio, OPT needs to serve as many requests as possible using the least total tour length possible.
OPT always combines $\capV$ requests from $Q_1$ with $\capV$ requests from $Q_2$ and serves them together by performing a subtour of length $\abs{\circuit}$. In addition, to avoid not fully loaded moves for OPT, the adversary choses 
$\abs{Q_1}=\abs{Q_2}$ and as a multiple of $\capV$ which leads to $OPT(\sigma)=\frac{\abs{G_\sigma}}{2\capV} \cdot \abs{\circuit}$, therefore
 \[
 \frac{SIF_L(\sigma)}{OPT(\sigma)}= \frac{\frac{\abs{G_\sigma}}{\capV} \cdot \abs{\circuit}}{\frac{\abs{G_\sigma}}{2\capV} \cdot \abs{\circuit}}= 2
\]
is the maximum possible ratio between $SIF_L(\sigma)$ and $OPT(\sigma)$ taken over all possible sequences on a circuit of length $\abs{\circuit}$ during the lunch.
\hfill$\square$
\end{proof}
\subsection{Elevator mode}
The elevator mode is a less restricted operation mode where one VIPA runs on a predefined line and can change its direction at any station of this line to move towards a requested station.
One end of this line is distinguished as origin $O$ (say, the ``left'' end).
\paragraph{Optimal offline solution:}
In order to obtain the optimal offline solution $OPT(\sigma)$ w.r.t. minimizing the total tour length, we compute a min cost flow in a suitable network. 
Given a line $\subline=(v_0,\dotsc,v_{\nsub})$ with origin $v_0$ as a subnetwork, capacity $\capV$ of a VIPA, and  a request sequence $\sigma$ with requests $r_j=(t_j,x_j,y_j,z_j)$.\\
In the sequel, we distinguish in which direction an edge $v_iv_{i+1}$ of $\subline$ is traversed and speak of the up arc $(v_i,v_{i+1})$ and the down arc $(v_{i+1},v_i)$.
In order to construct the network, we proceed as follows:
\begin{itemize}
 \item[$\bullet$]Neglect the release dates $t_j$ and only consider the loads of the requests $z_j$, their origins $x_j$ and their destinations $y_j$.
 \item[$\bullet$]Partition the requests into two subsets:
  \subitem $\uprequests$ of ``up-requests''  $r_j \in \sigma$ with $x_j<y_j$,
  \subitem $\downrequests$ of ``down-requests''  $r_j \in \sigma$ with $x_j>y_j$.
\item[$\bullet$] Determine the loads of all up arcs $(v_i,v_{i+1}))$ or down arcs $(v_{i+1},v_i)$  of the line $\subline$ as a weighted sum of the load of all request-paths $(x_j,y_j)$ containing this arc: 
\subitem $load(v_i,v_{i+1})=\sum_{(v_i,v_{i+1})\in(x_j,y_j),x_j<y_j} z_j$ $\forall i \in \{0,\nsub-1\}$, $\forall r_j \in \uprequests$ 
\subitem $load(v_{i+1},v_i)=\sum_{(v_{i+1},v_i)\in(x_j,y_j),x_j>y_j} z_j$ $\forall i \in \{0,\nsub-1\}$, $\forall r_j \in \downrequests$ 
\item[$\bullet$] Determine the ``multiplicities'' $m$ of all up/down arcs: 
in order to serve all the requests in $\sigma$, each arc $(v_i,v_{i+1})$ must be visited $m_{(i,i+1)}=\lceil \frac{load(v_i,v_{i+1})}{\capV} \rceil$ times and each arc $(v_{i+1},v_i)$ 
must be visited $m_{(i+1,i)}=\lceil \frac{load(v_{i+1},v_i))}{\capV} \rceil$ times. 
In case the multiplicity $m_{(i,i+1)}$ resp. $m_{(i+1,i)}$ is equal to zero, then the corresponding up resp. down arc is removed.
\end{itemize}
Now we build a network ~$G_E = (V_E, A_E)$, where 
\begin{itemize}
 \item[$\bullet$] the node set $V_E=V^{up(o)}\cup V^{up(d)} \cup V^{down(o)} \cup V^{down(d)} $ contains 
 \subitem the origin nodes of all up arcs where the multiplicity is different from zero in $V^{up(o)}$,
 \subitem the destination nodes of all up arcs where the multiplicity is different from zero in $V^{up(d)}$,
 \subitem the origin nodes of all down arcs where the multiplicity is different from zero in $V^{down(o)}$,
 \subitem the destination nodes of all down arcs where the multiplicity is different from zero in $V^{down(d)}$,
 \subitem the origin $v_0$ of the line $\subline$ as source $s$ and as sink $t$.
\item [$\bullet$]The arc set $A_E =A_s \cup A_U \cup A_D \cup A_L\cup A_t$ is composed of:
\subitem source arcs from the source $s$ to all $v_i^{up(o)}\in V^{up(o)}$ and all $v_i^{down(o)}\in V^{down(o)}$ in $A_{s}$,
 \subitem  up arcs $(v_i^{up(o)}, v_{i+1}^{up(d)})$ whenever $m_{(i,i+1)}\neq 0$ in $A_U$,
 \subitem  down arcs $(v_{i+1}^{down(o)},v_i^{down(d)})$ whenever $m_{(i+1,i)}\neq 0$ in $A_D$,
 \subitem link arcs in $A_L$ going from all $v_i^{up(d)}\in V^{up(d)}$ to all $v_i^{down(o)}\in V^{down(o)}$, and from all $v_i^{down(d)}\in V^{down(d)}$ to all $v_i^{up(o)}\in V^{up(o)}$,
 \subitem sink arcs from all $v_i^{up(d)}\in V^{up(d)}$ and from all $v_i^{down(d)}\in V^{down(d)}$ to the sink $t$ in $A_{t}$.
\end{itemize}
Accordingly, the objective function considers costs $d(a)=d(u,v)$ for the flow $f$ on all arcs $a=(u,v)$ in $A_E$, where $d(u,v)$ is the length of a shortest path from $u$ to $v$ in the line $\subline$.
To correctly initialize the system, we use the source node $s$ as source for the flow and the sink node $t$ as its destination.
For all internal nodes, we use normal flow conservation constraints.
We require a flow on all up/down arcs $f(a)=m(a)$ for all $a \in \{A_U\cup A_D\}$, see constraint~(\ref{eq: elevator: min-cost: flow subtour arcs: ilp}). 
We finally add the constraint~(\ref{eq: elevator: min-cost: subtour-elimination: ilp}) to eliminate all possible isolated cycles that the flow may contain (since the network contains directed cycles).

This leads to a Min-Cost Flow Problem, whose output is a subset of arcs needed to form a transportation schedule for a metric task system, whose tasks are induced by the requests.
The corresponding integer linear program is as follows:
\begin{subequations}
\label{eq: elevator: min-cost: model: ilp}
  \begin{align}
  \min \, & \sum_{a\in A_E} d(a)f(a)\\
  s.t. \, & \sum_{a\in \delta^-(s)}f(a) = 1\\
          & \sum_{a\in \delta^-(v)} f(a) = \sum_{a\in \delta^+(v)} f(a)  && \hspace{-0.5cm} \forall v \neq \{s,t\} \label{eq: elevator: min-cost: flow conservation: ilp: 10}\\
          & \sum_{a\in \delta(W)}f(a)  \geq 2 && \hspace{-0.5cm} \forall W \subset V_E \setminus \{s,t\} , 2 \leq |W| \leq |V_E|-3  \label{eq: elevator: min-cost: subtour-elimination: ilp}\\
          & f(a)=m(a) && \hspace{-0.5cm} \forall a \in \{A_U\cup A_D\} \label{eq: elevator: min-cost: flow subtour arcs: ilp}\\
          & f(a) \in \Z_+ \label{eq: elevator: min-cost: integrality: ilp}
  \end{align}
\end{subequations}
where $\delta^-(v,t)$ denotes the set of outgoing arcs of $v$, $\delta^+(v)$ denotes the set of incoming arcs of $v$ and $\delta(W)$ denotes the set of incoming or
outgoing arcs $(u,v)$ of $W$ s.t. $u\in W$ and $v\notin W$ or $u\notin W$ and $v\in W$.
The time required to compute the integer linear program grows in proportion to $2^{\abs{V}}$ due to the constraint~(\ref{eq: elevator: min-cost: subtour-elimination: ilp})
that eliminates all possible isolated cycles, and, hence, it may grow exponentially. 
However, this integer linear program can be computed in reasonable time provided that the number $V$ of nodes in the original network (the line) is small.
\begin{remark}
The family of constraints~(\ref{eq: elevator: min-cost: subtour-elimination: ilp}) can be generated and then each inequality is separated to verify if it is violated or not.
However, due to their exponential number, the process of separation in order to verify if a solution satisfies all constraints is exponential.
Since the number of subtour elimination constraints is exponential, we may firstly compute the integer linear program without the constraints~(\ref{eq: elevator: min-cost: subtour-elimination: ilp}). Then we check
if there is an isolated cycle in the solution obtained, if yes, we add only the constraint~(\ref{eq: elevator: min-cost: subtour-elimination: ilp}) using the nodes of this isolated cycle. 
This procedure is repeated until a solution without isolated cycles is found.
\end{remark}
Finally, the flow in the time-expanded network is interpreted as a transportation schedule.
The tracks of the VIPA over time can be recovered from the flow $f(a)$ on the arcs by standard flow decomposition, see~\cite{ahuja1993network}. 
Hereby, a flow $f(a)$ on an arc $a = (u, v)$ corresponds to a move of a VIPA on this arc. Based on the flow values, one can construct a unique path from source $s$ to sink $t$ traversing 
all arcs $a$ with positive flow exactly $f(a)$ times.
This shows that the optimal solution of system \eqref{eq: elevator: min-cost: model: ilp} corresponds to a transportation schedule with minimal total tour length
for the offline problem behind the elevator mode.

\begin{example}
 \label{exa: comp: opt: elevator: offline solution}
Consider a line $\subline=(v_0,\dotsc,v_{\nsub})$ with origin $v_0$, a unit distance between $v_i$ and $v_{i+1}$ for each $i$, with a set $\sigma$ of 9 requests shown in Figure~\ref{fig: opt-elevator-requests},
and a VIPA with capacity $\capV=2$.
The resulting network $G_E = (V_E, A_E)$ of the line presented in this example is illustrated in Figure~\ref{fig: opt: elevator: min flow}.
The optimal offline solution, the transportation schedule of the VIPA with a minimal total tour length, is obtained by computing the presented integer linear program for a min-cost flow problem in the constructed network $G_E = (V_E, A_E)$.

\begin{figure}[ht]
\centering
 \begin{tabular}{>{\centering\arraybackslash}p{3.5cm}|>{\centering\arraybackslash} p{2cm}|>{\centering\arraybackslash}p{2cm}|>{\centering\arraybackslash}p{2cm}|>{\centering\arraybackslash}p{2cm}}
 \multicolumn{5}{c }{\includegraphics{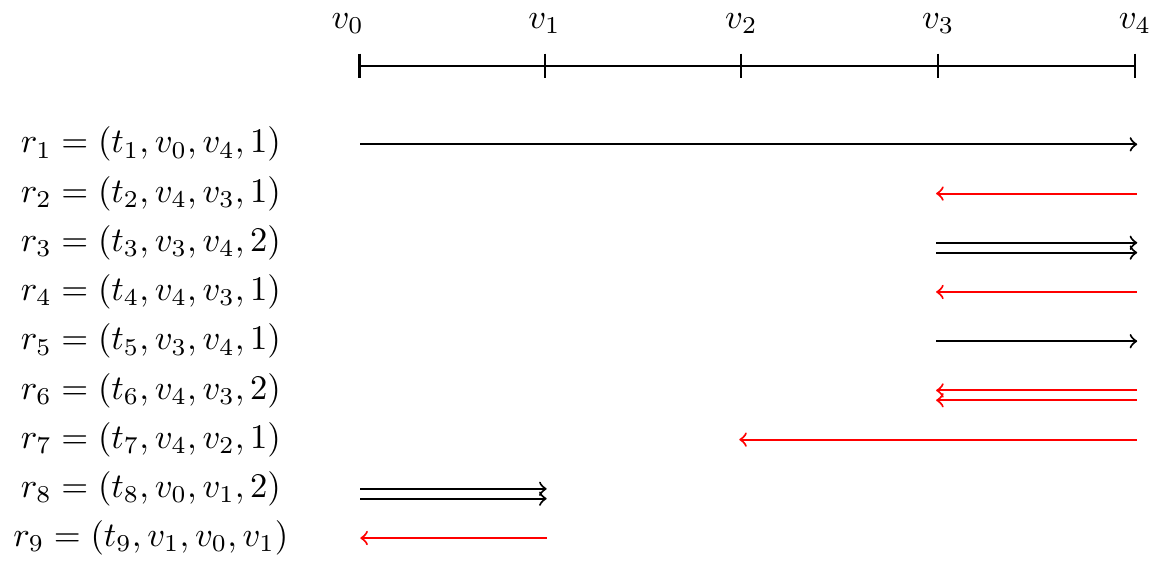}} \\
 \hline
 $load(up arcs)$&3&1&1&4\\
 \hline
 $load(down arcs)$&1&0&1&5\\
 \hline
 \hline
  $m(up arcs)=\lceil \frac{load}{\capV} \rceil$&2&1&1&2\\
 \hline
 $m(down arcs)=\lceil \frac{load}{\capV} \rceil$&1&0&1&3\\
 \hline
\end{tabular}
 \caption{This figure illustrates the line $\subline=(v_0,\dotsc,v_{\nsub})$ with origin $v_0$, and a set $\sigma$ of 9 requests and a VIPA with capacity $\capV=2$.
 The requests are partitioned into two subsets ``up-requests'' (arcs in black) and ``down-requests'' (arcs in red). Each arc represents a load of 1, for example $r_3$ is represented by 2 arcs going from $v_3$ to $v_4$.
 The loads of all up arcs (all arcs $(v_i,v_{i+1}))$ or down arcs (all arcs $(v_{i+1},v_i)$) of the line $\subline$ are shown in the first and second row of the table.
 Then the third and the forth rows contain  the ``multiplicities'' $m$ of all up/down arcs.}
 \label{fig: opt-elevator-requests}
 \end{figure}
\begin{figure}[ht]
    \centering
    \includegraphics[width=0.99\textwidth]{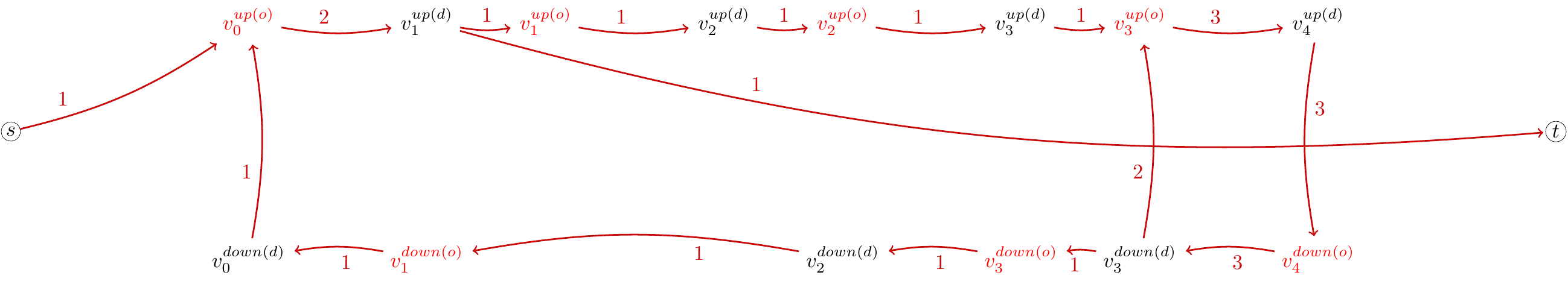}
 \caption{This figure illustrates the flow computed by the presented integer linear program for a min-cost flow problem in the network $G_E = (V_E, A_E)$ of Example~\ref{exa: comp: opt: elevator: offline solution}.
 The values above the arcs correspond to the value of the 
 flow $f(a)$ or the number of times the VIPA traverses the arc $a$ in the transportation schedule.}

 \label{fig: opt: elevator: min flow}
\end{figure}
\end{example}
\paragraph{Online algorithm:}
An algorithm MRIN (``Move Right If Necessary'') has been proposed for the Online-Traveling Salesman Problem (where no transports have to be performed, but only points to be visited)
on a line 
and analyzed w.r.t. minimizing 
the makespan \cite {IJC:BMK-2001}, giving a competitive ratio of $3/2$.
We generalize MRIN to the Pickup and Delivery Problem and analyze it w.r.t. minimizing the total tour length. 
In elevator mode, the server (VIPA) has the choice to continue its tour in the current direction, to wait at its current position or to change its driving direction.
Accordingly, we propose the algorithm MAIN. The adversary can again ``cheat'' the strategy of MAIN as the following example shows.
\begin{algorithm}
MAIN (``Move Away If Necessary'')\\
\SetKwInOut{Input}{Input}\SetKwInOut{Output}{Output}
\Input{ request sequence $\sigma$, line $\subline$ with origin $v_O$, $\capV$ }
\Output{tour on $\subline$ to serve all requests in $\sigma$ }
 \emph{initial server position $s:=v_O$}\;
 \emph{initial set of currently waiting requests (already released but not yet served requests) $\sigma^\prime:=\{r_j\in \sigma: t_j=0\}$}\;
  \While{ $\sigma^\prime \neq \varnothing$}{
  determine the subset $\sigma^\prime_{up}$ of requests $r_j=(t_j,x_j,y_j,z_j)\in \sigma^\prime $ with $s\leq x_j \leq y_j$\;
  \If{ $\sigma^\prime_{up} \neq \varnothing$}
  { Serve all requests (or up to $\capV$ passengers) in $\sigma^\prime_{up}$ (moving away from $v_O$ to furthest destination $y_k$ among all $r_j\in \sigma^\prime_{up}$)\;}
 \Else{
  determine subset $\sigma^\prime_{down}$ of requests $r_j=(t_j,x_j,y_j,z_j)\in \sigma^\prime $ with $x_j > y_j$\;
  serve all requests (or up to $\capV$ passengers) in $\sigma^\prime_{down}$ while moving to the origin\;
  }
  update $s$ and $\sigma^\prime$ (remove all served requests, add all newly released requests)\;
}
\end{algorithm}\DecMargin{1em}

\begin{example}
\label{exa: comp: MAIN: 2CAPL-totaltourlength}
Consider a line $\subline=(v_0,\dotsc,v_{\nsub})$ with origin $v_0$, a unit distance between $v_i$ and $v_{i+1}$ for each $i$, and one unit-speed server with capacity $\capV$.
The adversary releases  a sequence $\sigma$ of 
uniform requests that force MAIN to leave the origin of the line and perform a subtour of a certain distance for each request,
whereas the adversary is able to serve all requests in a single subtour of length $2\abs{L}$:\\
\begin{subequations}
\label{eq: comp: MAIN: CAPC-totaltourlength}
  \begin{align*}
 &\hspace{-0.75cm}\text{The first block } \sigma_1 \text{ of } \nsub \cdot \capV \text{ requests:}\\
 r_1 & = (0,v_0,v_1,1) \\
 r_j & =(t_{j-1}+2d(v_0,v_1),v_0,v_1,1) \text{ for } 2 \leq j \leq \capV\\
r_j & =(t_{j-1}+2d(v_0,v_1),v_1,v_2,1) \text{ for } j=\capV+1\\
r_j  & =(t_{j-1}+2d(v_0,v_2),v_1,v_2,1) \text{ for } \capV+2 \leq j \leq 2\capV\\
r_j & =(t_{j-1}+2d(v_0,v_{\nsub-1}),v_{\nsub-1},v_{\nsub},1) \text{ for } j=(\nsub-1)\capV+1\\
 r_j & =(t_{j-1}+2d(v_0,v_{\nsub}),v_{\nsub-1},v_{\nsub},1) \text{ for } (\nsub-1)\capV+2 \leq j \leq \nsub\capV \\
 &\hspace{-0.75cm}\text{The second block } \sigma_2 \text{ of } 2\kp \cdot \capV \text{ requests}\\
r_j & =(t_{j-1}+2d(v_0,v_{\nsub}),v_{\nsub},v_{\nsub-1},1) \text{ for } \nsub\capV+1 \leq j \leq (\nsub+1)\capV \\
r_j & =(t_{j-1}+2d(v_0,v_{\nsub}),v_{\nsub-1},v_{\nsub},1) \text{ for } (\nsub+1)\capV+1 \leq j \leq (\nsub+2)\capV \\
r_j & =(t_{j-1}+2d(v_0,v_{\nsub}),v_{\nsub},v_{\nsub-1},1) \text{ for } (\nsub+2\kp-2)\capV+1 \leq j \leq (\nsub+2\kp-1)\capV \\
 r_j & =(t_{j-1}+2d(v_0,v_{\nsub}),v_{\nsub-1},v_{\nsub},1) \text{ for } (\nsub+2\kp-1)\capV+1 \leq j \leq (\nsub+2\kp)\capV \\
   \end{align*}
\end{subequations}
\begin{subequations}
  \begin{align*}
 &\hspace{-0.75cm}\text{The third block } \sigma_3 \text{ of }\nsub \cdot \capV \text{ requests:}\\
r_j & =(t_{j-1}+2d(v_0,v_{\nsub}),v_{\nsub},v_{\nsub-1},1) \text{ for } (\nsub+2\kp+1)\capV+1 \leq j \leq (\nsub+2\kp+2)\capV \\
r_j & =(t_{j-1}+2d(v_0,v_{\nsub}),v_{\nsub-1},v_{\nsub-2},1) \text{ for } j=(\nsub+2\kp+2)\capV+1\\
r_j & =(t_{j-1}+2d(v_0,v_{\nsub}),v_{\nsub},v_{\nsub-1},1) \text{ for } (\nsub+2\kp+2)\capV+2 \leq j \leq (\nsub+2\kp+3)\capV \\
r_j & =(t_{j-1}+2d(v_0,v_2),v_1,v_0,1)\text{ for } j=(2\nsub+2\kp-1)\capV+1\\
r_j & =(t_{j-1}+2d(v_0,v_1),v_1,v_0,1) \text{ for }  (2\nsub+2\kp-1)\capV+2  \leq j \leq (2\nsub+2\kp)\capV\\
  \end{align*}
\end{subequations}
MAIN starts its VIPA at time $t=0$ to serve $r_1=(0,v_0,v_1,1)$ and finishes the first subtour of length $2d(v_0,v_1)=2$ without serving any further request.
When the VIPA operated by MAIN is back to the origin $v_0$, the second request $r_2=(2,v_0,v_1,1)$ is released and MAIN starts at $t=2$ a second subtour of length $2$ to serve $r_2$,
without serving any further request in this subtour. This is repeated for each request until serving the first block of $\nsub \cdot \capV$ requests 
yielding  \[MAIN(\sigma_1)=2\cdot \capV \sum_{1 \leq i \leq \nsub} i =\capV \cdot \abs{\subline} \cdot (\abs{\subline}+1) \]
Then at $t=t_{\nsub\capV+1}$ MAIN starts to serve $r_{\nsub\capV+1}$ from $v_n$ to $v_{n-1}$ and performs a subtour of length $2d(v_0,v_{\nsub})=2\abs{L}$. When the VIPA operated by MAIN is back
to the origin $v_0$, the request $r_{\nsub\capV+2}$ is released and MAIN performs a new subtour of length $2\abs{L}$ to serve it. This is repeated for each request
until serving the second block of $\kp \cdot2 \capV$ requests yielding  \[ MAIN(\sigma_2)= 2\kp \cdot 2\capV \abs{L}. \]
Finally in order to serve the third block MAIN has the same behavior as to serve the first block of requests yielding 
\[MAIN(\sigma_3)=2\cdot \capV \sum_{1 \leq i \leq \nsub} i=\capV \cdot \abs{\subline} \cdot (\abs{\subline}+1). \]
Therefore \[MAIN(\sigma)=\capV \cdot \abs{\subline} \cdot (\abs{\subline}+1 )+ 2\kp \cdot 2\capV \abs{L}= (\abs{\subline} + 1 + 2\kp) \cdot 2\abs{L}\cdot\capV.  \]
The adversary waits at the origin $v_0$ until $t=t_{\capV}$ and serves all requests $r_1,\dotsc,r_{\capV}$ from $v_0$ to $v_1$. Then he waits until $t=t_{2\capV}$ at $v_1$ and serves all requests 
$r_{\capV+1},\dotsc, r_{2\capV}$ from $v_1$ to $v_2$. This is repeated for all $\capV$ requests from $v_i$ to $v_{i+1}$ until the adversary arrives to $v_{\nsub}$. OPT served the first block of $\nsub \cdot \capV$
requests with a total tour length equal to $\abs{L}$.
Then the adversary begins to oscillate his VIPA between $v_{\nsub}$ and $v_{\nsub-1}$ and serves each time $\capV$ requests, this is repeated $2\kp$ times leading to a total tour length for $\sigma_2$ equal to $2\kp$.
Finally the adversary follows the other direction and waits each time until $\capV$ requests are released to serve them,
for all $\capV$ requests from $v_i$ to $v_{i-1}$ until reaching $v_0$, yielding 
$OPT(\sigma)=2\abs{L} +2\kp.$
Therefore, we obtain 
\[
 \frac{MAIN(\sigma)}{OPT(\sigma)}= \frac {2\cdot \capV \cdot\abs{\subline}\cdot (\abs{\subline}+1)+(\kp \cdot 2\cdot \capV \cdot2\abs{\subline})}{2(\abs{\subline} +\kp)}
 =\frac {2\cdot \capV \cdot\abs{\subline}\cdot (\abs{\subline}+1+2\kp)}{2(\abs{\subline} +\kp)}\]
 \[= \capV \cdot\abs{\subline}+ \frac{(1+\kp)}{\abs{\subline}+\kp}\capV \cdot\abs{\subline}
 \underset{\kp\to +\infty}{\longrightarrow} 2\capV \cdot \abs{\subline}
\]
as a lower bound for the competitive ratio of MAIN.
\begin{figure}[ht]
    \centering
    \includegraphics[width=0.99\textwidth]{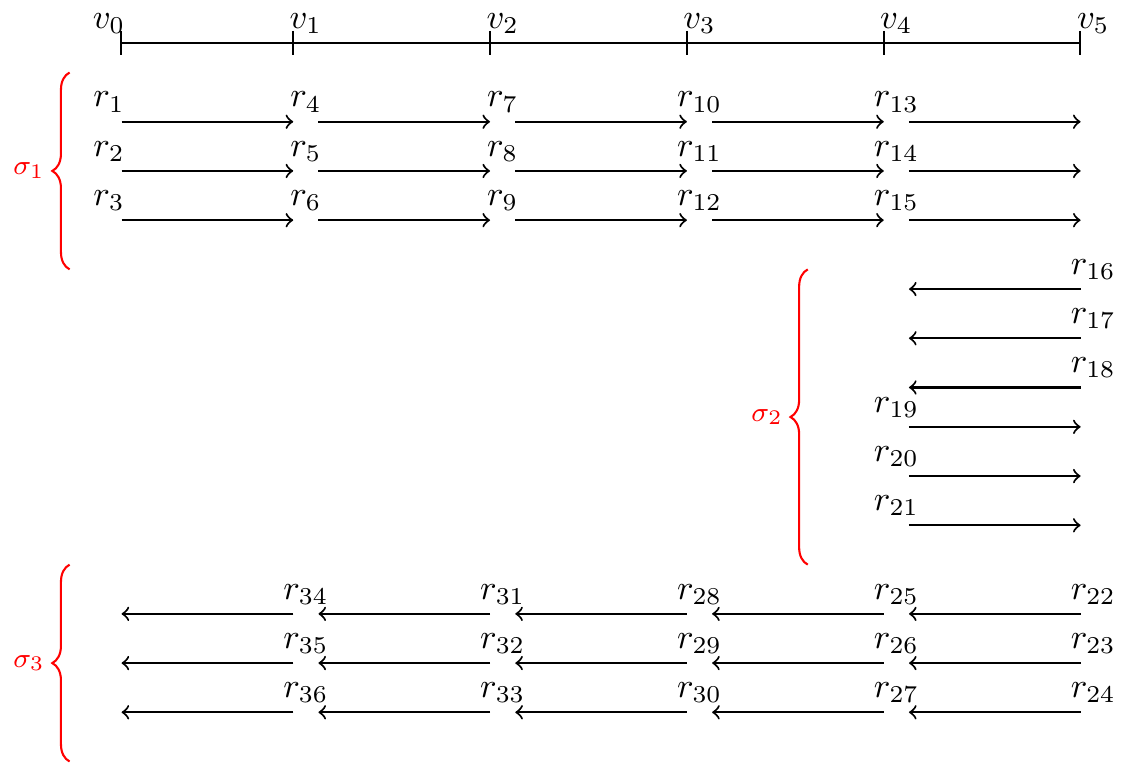}
 \caption{This figure illustrates the set $\sigma=\sigma_1\cup\sigma_2\cup \sigma_3$ of requests (arcs under the line $\subline=(v_0,\dotsc,v_5)$ with origin $v_0$) 
 from Example~\ref{exa: comp: MAIN: 2CAPL-totaltourlength} for $\capV = 3, \nsub=5$ and $\kp=1$.}
 \label{fig: main-worst-case-requests}
\end{figure}
\end{example}

 We can determine an upper bound for the competitive ratio of MAIN close to the ratio obtained by the previous example:
\begin{theorem}\label{th-MAIN-cap-general}
MAIN is $2\capV \cdot \abs{\subline}$-competitive w.r.t minimizing the total tour length for one VIPA operating in elevator mode on a line $\subline$ with length $\abs{\subline}$. 
\end{theorem}
\begin{proof}
The worst transportation schedule results if all requests are uniform and the VIPA operated by MAIN performs a separate subtour serving a single request $r_j=(t_j,x_j,y_j,1)$ each time the VIPA leaves the origin $v_0$ of the line,
yielding $\sum_{ r_j \in \sigma}2\max(d(v_0,x_j),d(v_0,y_j))$ 
as total tour length. \\
To maximize the ratio between the total tour length obtained by MAIN and the optimal offline solution, we need to ensure that
 \begin{itemize} 
 \item we do not have a move with a load less than the capacity $\capV$ of the VIPA in the transportation schedule of $OPT$;
 \item all requests in $\sigma$ can be served with as few and as short subtours as possible in OPT. 
\end{itemize}
The worst ratio of subtours can be obtained when
\begin{itemize}
 \item OPT oscillates fully loaded between two neighbored nodes of $\subline$,
 \item MAIN is forced to traverse the whole line twice per passenger, i.e. oscillates between $v_0$ and $v_\nsub$.
\end{itemize}
For that, $v_\nsub$ needs to be either origin or destination of each request, and the delay between the release dates needs to be sufficiently large. 
 This can be achieved with subsequence $\sigma_2$ from Example~\ref{exa: comp: MAIN: 2CAPL-totaltourlength}, with $\kp$ blocks each consisting of
 \begin{itemize}
  \item $\capV$ consecutive uniform requests from $v_n$ to $v_{n-1}$, alternated by
  \item $\capV$ consecutive uniform requests from $v_{n-1}$ to $v_n$,
 \end{itemize}
 always with a delay $2\abs{\subline}$ between the release dates of any two requests $r_j$ and $r_{j+1}$.
 We obtain $OPT(\sigma_2)=2\kp$ and $MAIN(\sigma_2)=\kp\cdot2\cdot\capV\cdot2\cdot\abs{\subline}$ which leads to the studied subtour ratio of
 \[
  \frac{MAIN(\sigma_2)}{OPT(\sigma_2)}=\frac{\kp\cdot2\cdot\capV\cdot2\cdot\abs{\subline}}{2\kp}=2\cdot\capV\cdot\abs{\subline}.
 \]
 However, this ratio so far neglects the initial and final server position $v_0$ for the VIPA operated by OPT. The requirement of starting and ending the tour in $v_0$ leads to a total tour length for OPT of
 \[
  OPT(\sigma)=\abs{\subline}\cdot2\kp\cdot \abs{\subline}.
 \]
 In order to maximize the ratio of the complete tours, the adversary releases more requests to ensure that the VIPA operated by
 \begin{itemize}
  \item OPT can arrive at $v_n$ (resp. return from $v_n$ to $v_0$) fully loaded on each arc,
  \item MAIN is forced to oscillate between $v_0$ and the destination $y_j$ (resp. the origin $x_j$) of each uniform request $r_j$.
 \end{itemize}
 This can be achieved with the subsequences $\sigma_1$ and $\sigma_3$ from Example~\ref{exa: comp: MAIN: 2CAPL-totaltourlength} with
 \begin{itemize}
  \item $\capV$ consecutive uniform requests from $v_i$ to $v_{i+1}$ for each $0\leq i< \nsub$ and
  \item $\capV$ consecutive uniform requests from $v_i$ to $v_{i-1}$ for each $\nsub\geq i\geq 1$,
 \end{itemize}
 always with delay $2\cdot d(v_0,y_j)$ resp. $2\cdot d(x_j,v_0)$ between the release dates of any two requests $r_j$ and $r_{j+1}$ within these subsequences. We obtain (as in Example~\ref{exa: comp: MAIN: 2CAPL-totaltourlength}) that
 \[MAIN(\sigma_1)=MAIN(\sigma_3)=2\cdot \capV \sum_{1 \leq i \leq \nsub} i =\capV \cdot \abs{\subline} \cdot (\abs{\subline}+1). \]
This finally leads to
\[
 \frac{MAIN(\sigma)}{OPT(\sigma)}= \frac {2\cdot \capV \cdot\abs{\subline}\cdot (\abs{\subline}+1)+(\kp \cdot 2\cdot \capV \cdot2\abs{\subline})}{2(\abs{\subline} +\kp)}
 =\frac {2\cdot \capV \cdot\abs{\subline}\cdot (\abs{\subline}+1+2\kp)}{2(\abs{\subline} +\kp)}\]
 \[=\capV \cdot\abs{\subline}+ \frac{(1+\kp)}{\abs{\subline}+\kp}\capV \cdot\abs{\subline}\underset{\kp\to +\infty}{\longrightarrow} 2\capV \cdot \abs{\subline}
\]
as the maximum possible ratio between $MAIN(\sigma)$ and $OPT(\sigma)$ taken over all possible sequences on a line $\subline$. 
\hfill$\square$
\end{proof}

Concerning the lunch scenario, we may consider VIPAs operating in elevator mode on lines, where each line has the restaurant as its distinguished origin. 
A sequence $\sigma'$ containing the first $\capV$ requests of the first block $\sigma_1$ and the last $\capV$ requests from the third block $\sigma_3$ from the sequence 
presented in Example~\ref{exa: comp: MAIN: 2CAPL-totaltourlength} shows that $2 \cdot \capV$ is a lower bound on the competitive ratio of MAIN.
As for the morning resp. evening scenario, we may consider VIPAs operating in elevator mode on lines, where each line has a parking as its distinguished origin. A sequence $\sigma'$ containing the first $\capV$ requests 
of the first block $\sigma_1$ resp. the last $\capV$ requests from the third block $\sigma_3$ from the sequence 
presented in Example~\ref{exa: comp: MAIN: 2CAPL-totaltourlength} shows that $\capV$ is a lower bound on the competitive ratio of MAIN.
We can show that these examples are the worst cases for MAIN during lunch, morning and evening:

\begin{theorem}\label{th-MAIN-cap-lunch}
For one VIPA with capacity $\capV$ operating in elevator mode on a line, MAIN is w.r.t. the objective of minimizing the total tour length
\begin{itemize}
 \item $2 \cdot \capV$-competitive during the lunch scenario,
 \item $\capV$-competitive during the morning resp. the evening scenario. 
\end{itemize}
\end{theorem}

\begin{proof}
The worst transportation schedule results if the VIPA operated by MAIN performs a separate subtour serving a single uniform request $r_j=(t_j,v_0,v_1,1)$ or  $r_j=(t_j,v_1,v_0,1)$
each time the VIPA leaves the origin $v_0$ of the line,
  yielding $\sum_{ r_j \in \sigma}2d(v_0,v_1)$ as total tour length. 
MAIN can indeed be forced to show this behavior by releasing the requests accordingly (i.e. by using requests with $z_j=1$ each and with sufficiently large delay between 
  $t_j$ and $t_{j+1}$).
In order to maximize the ratio between the total tour length obtained by MAIN and the optimal offline solution, we need to ensure that 
 \begin{itemize} 
 \item[$\bullet$] we do not have a move from or to the origin with a load less than the capacity $\capV$ of the VIPA in the transportation schedule of $OPT$. For that, the adversary releases
 \subitem during the lunch $\capV$ many requests traversing the same arc.
Whereas MAIN traverses $d(v_0,v_1)$ twice to serve a request $r_j=(t_j,v_0,v_1,1)$ or  $r_j=(t_j,v_1,v_0,1)$, OPT travels $d(v_0,v_1)$ once to serve the request and can share it
with $\capV -1$ others.
\subitem during the morning/evening $\capV$ many requests traversing the same arc.
Whereas MAIN traverses $d(v_0,v_1)$ twice to serve a request $r_j=(t_j,v_0,v_1,z_j)$ resp.  $r_j=(t_j,v_1,v_0,z_i)$, OPT travels the same distance to serve the request but can share it
with $\capV -1$ others.
 \item[$\bullet$] all requests in $\sigma$ can be served with as few and as short subtours as possible in OPT. For that, the adversary releases
 \subitem during the lunch a sequence $\sigma$ of $2\capV$ requests:
 $\capV$ many requests $v_0 \to v_1$ followed by
 $\capV$ many requests $v_1 \to v_0$. Therefore we obtain 
  \[ MAIN(\sigma)=\sum_{r_j\in\sigma}2d(v_0,v_1)=2\cdot \capV \cdot 2d(v_0,v_1) \text{ and } OPT(\sigma)=2d(v_0,v_1)\]
\[\text{s.t. } \frac{ MAIN(\sigma)}{OPT(\sigma)}=\frac{2\cdot \capV \cdot 2d(v_0,v_1)}{2d(v_0,v_1)}=2\capV\]
is the maximum possible ratio between $MAIN(\sigma)$ and $OPT(\sigma)$ taken over all possible sequences on a line during the lunch.
\subitem during the morning/evening a sequence $\sigma$ of $\capV$ requests:
$\capV$ many requests $v_0 \to v_1$ resp. $\capV$ many requests $v_1 \to v_0$. Therefore we obtain 
  \[ MAIN(\sigma)=\sum_{r_j\in\sigma} 2d(v_0,v_1)=\capV \cdot 2d(v_0,v_1) \text{ and } OPT(\sigma)=2d(v_0,v_1)\]
\[\text{s.t. } \frac{ MAIN(\sigma)}{OPT(\sigma)}=\frac{\capV \cdot 2d(v_0,v_1)}{2d(v_0,v_1)}=\capV\]
is the maximum possible ratio between $MAIN(\sigma)$ and $OPT(\sigma)$ taken over all possible sequences on a line during the morning resp. evening.
\hfill$\square$
\end{itemize}
\end{proof}

\subsection{Computational results}
This section deals with computational experiments for the proposed online algorithms.
In fact, due to the very special request structures of the previously presented worst case instances, we can expect a better behavior of the proposed online algorithms in average.
The computational results presented in Table~\ref{tab: computational results_tram} and~\ref{tab: computational results_elevator}  support this expectation, they compare the total tour length $TTL$ computed by the online algorithms
with the optimal offline solution. 
The computations use instances based on the network from the industrial site of Michelin and randomly generated request sequences resembling typical instances that occurred during the 
experimentation \cite{RFIA2016}.The computations are performed with the help of a simulation tool developed by Yan Zhao~\cite{ZHAO}.
The instances use subnetworks as a circuit or a line depending on the mode and the algorithm 
with 1 to 5 VIPAs, 5-200 requests, 1-12 as the maximum load $z_j$ of a request.
For every parameter set we created 6 test instances. In these tables, the instances are grouped by the number of requests and the capacity of the VIPA. The average results of the instances are shown. 
The operating system for all tests is Linux (CentOS with kernel version 2.6.32).
The algorithms SIR, $SIF_L$, MAIN and $\OPTTRAM$ have been implemented in Java.
For solving the integer linear program to get the optimal solution for the elevator mode, we use Gurobi 8.21.

\begin{table}[ht]
\centering
 \caption{This table shows the computational results for several test instances of the algorithms $SIR$ and $SIF_L$
    in comparison to the value of the optimal solution. In this table, the instances are grouped by the number of requests (1st column) and the capacity (2nd column). 
    Furthermore the values of the total tour length $TTL$ found by SIR with morning, evening and lunch instances are shown in comparison with the optimal solution  $OPT$ 
    and the ratio between them. Then the values of the total tour length $TTL$ found by $SIF_L$ with lunch instances are shown in comparison with the optimal solution $OPT$
    and the ratio between them. Finally, the results for instances respecting the general scenario are shown: the total tour length $TTL$ found by $SIR$, the optimal solution $OPT$
    and the ratio between them. The competitive ratio $c$ is shown for each of the scenarios, it is always greater than $\frac{TTL}{OPT}$ unless $c=\frac{TTL}{OPT}=1$. 
    In these test instances the length of the circuit used is equal 
    to $\abs{\circuit} =25$}
 \label{tab: computational results_tram}
{\scriptsize 
\begin{tabular}{|c|c|c|c|c|c|c|c|c|c|c|c|c|c|c|c|c|c|}
\hline
\multicolumn{2}{|c |}{ } & \multicolumn{3}{c |}{SIR (morning)} & \multicolumn{3}{c |}{SIR (Evening)}&\multicolumn{4}{c |}{SIR (Lunch)}& \multicolumn{2}{c |}{$SIF_L$ (Lunch)} &\multicolumn{4}{c| }{SIR (General)} \\
\multicolumn{2}{|c |}{ } & \multicolumn{3}{c |}{$c=\capV$} & \multicolumn{3}{c |}{$c=\capV$}&\multicolumn{4}{c |}{$c=2\capV$}& \multicolumn{2}{c |}{$c=2$} &\multicolumn{4}{c| }{$c=\capV\cdot\abs{\circuit}$ } \\
\hline
$\nreq$& $\capV$& $TTL$& $OPT$& $\frac{TTL}{OPT}$& $TTL$& $OPT$& $\frac{TTL}{OPT}$& $TTL$& $OPT$& $\frac{TTL}{OPT}$&$c$& $TTL$& $\frac{TTL}{OPT}$& $TTL$& $OPT$& $\frac{TTL}{OPT}$&$c$ \\
\hline
5&1&125&125&1,00&125&125&1,00&112,5&93,75&1,20&2&100&1,07&143,75&106,25&1,35&25\\
5&5&50&25&2,00&100&25&4,00&75&25&3,00&10&25&1,00&131,25&25&5,25&125\\
5&10&50&25&2,00&81,25&25&3,25&68,75&25&2,75&20&25&1,00&87,5&25&3,50&250\\
20&1&500&500&1,00&500&500&1,00&418,75&337,5&1,24&2&400&1,19&593,75&337,5&1,76&25\\
20&5&125&100&1,25&275&100&2,75&206,25&75&2,75&10&81,25&1,08&287,5&75&3,83&125\\
20&10&106,25&50&2,13&275&50&5,50&206,25&43,75&4,71&20&56,25&1,29&268,75&50&5,38&250\\
200&1&5000&5000&1,00&5000&5000&1,00&4300&2818,75&1,53&2&2968,75&1,05&5400&2706,25&2,00&25\\
200&5&1025&1000&1,03&2881,25&1000&2,88&2125&562,5&3,78&10&600&1,07&2600&550&4,73&125\\
200&10&637,5&500&1,28&2543,75&500&5,09&1512,5&275&5,50&20&468,75&1,70&2168,75&275&7,89&250\\
\hline
\end{tabular}}
\end{table}

\begin{table}[ht]
\centering
 \caption{This table shows the computational results for several test instances of the algorithm $MAIN$ in comparison to the value of the optimal solution.
 In this table, the instances are grouped by number of requests (1st column) and the capacity (2nd column). 
 First the values of the total tour length $TTL$  found by $MAIN$ for the instances respecting the morning scenario are shown, the optimal solution  $OPT$  and the ratio between them.
 Then the results for instances respecting the evening scenario are shown: the total tour length $TTL$ found by $MAIN$ and, the optimal solution $OPT$
 and the ratio between them.
 Finally, the values of the total tour length $TTL$ found by MAIN with general instances are shown in comparison with the optimal solution $OPT$ and the ratio between them. The competitive ratio $c$ is
 shown for each of the scenarios, it is always greater than $\frac{TTL}{OPT}$ unless $c=\frac{TTL}{OPT}=1$. 
 In these test instances the length of the line used is equal to $\abs{\subline} =32$}
 \label{tab: computational results_elevator}
{\scriptsize 
\begin{tabular}{|c|c|c|c|c|c|c|c|c|c|c|c|c|c|c|c|}
\hline
\multicolumn{2}{|c |}{ } & \multicolumn{3}{c |}{MAIN (Morning)}&\multicolumn{3}{c |}{MAIN (Evening)}&\multicolumn{4}{c| }{MAIN (Lunch)}& \multicolumn{4}{c| }{MAIN (General)} \\
\multicolumn{2}{|c |}{ } & \multicolumn{3}{c |}{$c=\capV$}&\multicolumn{3}{c |}{$c=\capV$}&\multicolumn{4}{c| }{$c=2\capV$}& \multicolumn{4}{c| }{$c=2\capV\cdot\abs{\subline}$} \\
\hline
$\nreq$& $\capV$& $TTL$& $OPT$& $\frac{TTL}{OPT}$ &  $TTL$& $OPT$& $\frac{TTL}{OPT}$ & $TTL$& $OPT$& $\frac{TTL}{OPT}$&$c$& $TTL$& $OPT$& $\frac{TTL}{OPT}$&$c$ \\
\hline
5&1&96&96&1,00&103,5&103,5&1,00&49,5&49,5&1,00&2&57&35&1,63&64\\
5&5&50,5&24,5&2,06&77&27&2,85&48&27&1,78&10&60&27&2,22&320\\
5&10&54&28&1,93&97&28&3,46&64,5&27&2,39&20&63,5&27&2,35&640\\
20&1&319&326&0,98&352&359&0,98&225,5&177&1,27&2&211&172&1,23&64\\
20&5&134,5&72,5&1,86&315&72,5&4,34&194,5&57,5&3,38&10&186,5&43&4,34&320\\
20&10&107&46,5&2,30&316&41,5&7,61&160,5&34&4,72&20&216&33,5&6,45&640\\
200&1&3253,5&3393,5&0,96&3304,5&3426,5&0,96&2240&1855,5&1,21&2&1713&1435,5&1,19&64\\
200&5&1049&693,5&1,51&3317,5&695&4,77&1560&368,5&4,23&10&1171&292&4,01&320\\
200&10&859,5&361&2,38&3359,5&355&9,46&1679&185&9,08&20&1085&152,5&7,11&640\\
\hline
\end{tabular}
}
\end{table}
\section{ Concluding Remarks} 

Vehicle routing problems integrating constraints on autonomy are new in the field of operational research but are important for the future mobility.
Autonomous vehicles, which are intended to be used as a fleet in order to provide a transport service, need to be effective also considering to their management.
We summarize the results of competitive analysis presented in this paper in Table~\ref{tab:  competitive results}.
\begin{table}[ht]
\centering
 \caption{This table shows the competitive ratios obtained in this paper. For each case we show the competitive ratio for the algorithm on which type of graphs, 
 using which mode. 
}
 \label{tab: competitive results}
{\scriptsize 
\begin{tabular}{|c|c|c|c|c|c|c|c|c|}
\hline
Mode (type of graph) & Algorithm & General& Lunch &Morning& Evening \\
\hline
Tram mode (Circuit)& SIR & $\capV \cdot|C|$ & $2\capV$ & $\capV$ & $\capV$\\
Tram mode (Circuit)& $SIF_M$ & - & - & 1& -\\
Tram mode (Circuit)& $SIF_E$ & - & -& - & 1  \\
Tram mode (Circuit)& $SIF_L$ & - & 2& - & -  \\
Elevator mode (Line)& MAIN & $2\capV\cdot|L|$ & $2\capV$ & $\capV$ & $\capV$ \\
\hline
\end{tabular}}
\end{table}
Competitive analysis has been one of the main tools for deriving worst-case bounds on the performance of algorithms but an online algorithm having the best competitive ratio in
theory may reach the worst case more frequently in practice with a certain topology.
That is the reason why we are not only  interested in the ``worst-case'' but also in the ``best-case''  performance of the algorithms, 
thus we need to determine properties which govern the behavior of each chosen algorithm and define the cases where it can be applied and give the best results in terms of performance.
So far, we can suggest the following:
\begin{itemize}
\item Morning/evening: partition the network  into disjoint circuits as subnetworks such that each subnetwork contains one parking $p$, assign one or several 
VIPA to every circuit operating in tram mode using $SIF_M$ resp. $SIF_L$.
\item Lunch time: consider a collection of circuits all meeting in a central station (the restaurant), one or several VIPAs on each circuit operating in tram mode using $SIF_L$.
\item in general: consider a spanning collection of lines and circuits meeting in a central station where one VIPA (in elevator mode) operates on each line using $MAIN$, one or several VIPAs (in tram mode) on each circuit using $SIR$.
 \end{itemize}
The future works are to analyze the proposed scenarios and algorithms further, e.g., by
\begin{itemize}
\item providing competitive ratios w.r.t. the objective of minimizing the makespan,
\item studying also the quality of service aspect by minimizing the waiting time for customers,
\item providing solution approaches for the taxi mode.
\end{itemize}

\bibliographystyle{hplain}
\providecommand{\availatURL}[1]{\ignorespaces \footnote{Avail. at URL
  \texttt{#1}}} \providecommand{\NP}{\textsf{NP}}
\providecommand{\bysame}{\leavevmode\hbox to3em{\hrulefill}\thinspace}
\providecommand{\MR}{\relax\ifhmode\unskip\space\fi MR }
\providecommand{\MRhref}[2]{%
  \href{http://www.ams.org/mathscinet-getitem?mr=#1}{#2}
}
\providecommand{\href}[2]{#2}
\bibliography{bib2015}

\end{document}